\documentclass[reqno,english]{amsart}
\usepackage{amsmath}
\usepackage{mathtools}
\usepackage{comment}
\usepackage[english]{babel}
\usepackage[
backend=bibtex,
style=numeric,
]{biblatex}
\newtheorem{thm}{Theorem}[section]
\newtheorem{prop}[thm]{Proposition}
\newtheorem{lem}[thm]{Lemma}
\newtheorem{cor}[thm]{Corollary}

\theoremstyle{definition}

\theoremstyle{remark}

\numberwithin{equation}{section}

\newcommand{\R}{\mathbb{R}}

\newcommand{\ca}[1]{\mathcal{#1}}
\newcommand{\mrm}[1]{\mathrm{#1}}

\newcommand{\dr}{\mathrm{d}}
\newcommand{\ve}{\varepsilon}
\newcommand{\de}{\delta}
\newcommand{\De}{\Delta}
\newcommand{\ga}{\gamma}
\newcommand{\iy}{\infty}
\newcommand{\p}{\partial}
\newcommand{\ti}{\times}

\newcommand{\ot}{\otimes}
\newcommand{\na}{\nabla}

\newcommand{\abs}[1]{\left |#1\right |}
\newcommand{\norm}[1]{\left\|#1\right\|}

\newcommand{\pa}[1]{\left( #1 \right)}

\newcommand{\T}{\mathbb{T}}
\newcommand{\eq}[1]{\begin{align} #1 \end{align} }
\newcommand{\mat}[1]{\begin{pmatrix} #1 \end{pmatrix}}
\newcommand{\eqq}[1]{\begin{align*} #1 \end{align*}}
\newcommand{\xra}[1]{\xrightarrow{#1}}
\newcommand{\Ga}{\Gamma}

\title{Quantitative Propagation of Chaos for the Mixed-Sign Viscous Vortex Model on the Torus}

\date{\today}

\author{Dominic Wynter}
\email{dlw49@cam.ac.uk}
\thanks{Department of Pure Mathematics and Mathematical Sciences, University of Cambridge, UK}

\calclayout

\newcommand{\Td}{{\mathbb{T}^d}}
\newcommand{\Rd}{{\mathbb{R}^d}}
\newcommand{\HN}{\mathcal{H}_N}

\begin{document}

\maketitle

\begin{abstract}
	We derive a quantiative propagation of chaos result for a mixed-sign point vortex system on $\T^2$ with independent Brownian noise, at an optimal rate. We introduce a pairing between vortices of opposite sign, and using the vorticity formulation of  2D Navier-Stokes, we define an associated \emph{tensorized} vorticity equation on $\T^2\ti\T^2$ with the same well-posedness theory as the original equation. Solutions of the new PDE can be projected onto solutions of Navier-Stokes, and the tensorized equation allows us to exploit existing propagation of chaos theory for identical particles.
\end{abstract}

\section{Introduction}

The point vortex model in 2D fluid dynamics is of  great interest in computational fluid dynamics \cite{Leonard1980} \cite{ChorinBernard1973} and in models of turbulence and chaos \cite{Ottino1989}, as well as in dynamical systems \cite{Aref2007}. A long standing problem has been to show that this model suitably approximates incompressible flow, but the non-Lipschitz nature of the Biot-Savart kernel has made this a delicate problem. We recall the vorticity form of Navier Stokes in two dimensions
\eq{\label{NS}
	\p_t \omega + u\cdot\na\omega = \nu\De\omega,\quad u = K*\omega
} which we obtain as the curl of the velocity equation, where we recover the velocity $u$ using the \emph{Biot-Savart }kernel $K$. On $\R^2$, this kernel is given as
\eq{
	K(x) = \frac{1}{2\pi}\frac{x^\perp}{\abs x^2}
} where we write $x^\perp = Jx = \begin{psmallmatrix}0&-1\\1&0\end{psmallmatrix}x$. On $\T^2$, of course, this has to be periodized, and we obtain the kernel 
\eq{
	K(x) = \frac{1}{2\pi}\frac{x^\perp}{\abs x^2}+\lim_{R\to\iy}\sum_{0<\abs k\le R} \frac{1}{2\pi}\frac{(x-k)^\perp}{\abs{x-k}^2},\qquad x\in[-1/2,1/2]^2,
} which can be checked to converge in $C^\iy$ \cite{Schochet1996}.
This kernel is therefore expressible as
\eqq{
	K(x) = \frac{1}{2\pi}\frac{x^\perp}{\abs x^2} + \ga(x),\quad x\in [-1/2-\ve,1/2+\ve]^2
} where $\ga:[-1/2-\ve,1/2+\ve]^2\to\R^2$ is a smooth correction. 

We now introduce the \emph{point vortex} model. In the inviscid case where $\nu=0$, if we take the initial vorticity $\omega_0 = \sum_{i=1}^N \alpha_i \de_{X_i}$ to be a discrete measure, then we would formally expect the solution to \eqref{NS} with this initial data to follow $\omega(t) = \sum_{i=1}^N \alpha_i \de_{X_i(t)}$ where the $X_i$ solve the \emph{point vortex} ODE
\eq{\label{pointvortexmodel}
	\dot X_i(t) = \sum_{j\ne i} \alpha_j K(X_i - X_j).
} 
This model is in fact very old, and dates back to the 19th century. It was introduced by Helmholtz \cite{Helmholtz1858}, and was studied by Poincaré as an early model for turbulent flow \cite{Poincare1893}, where he viewed them as vortex tubes in three dimensions. The first attempt to use the point vortex model for simulations appears to be by Rosenhead \cite{Rosenhead1931}, who in 1931 used it to approximate vortex sheets, however this model suffered from poor numerical stability \cite{Takami1964}. This problem seemed to be improved \cite{ChorinBernard1973} by using a cutoff Biot-Savart kernel  $K_\ve$, defined so that $$K_\ve(x) = K(x)$$ for $\abs x>\ve$, known as the \emph{vortex blob} model. Subsequent numerical research studied mainly this model \cite{Chorin1973}, and the first proofs for convergence of the vortex model to Euler used vortex blobs for smooth flows \cite{Hald1978}. These results were later extended to the original singular kernel \cite{Goodman-Hou-Louwengrub1990} \cite{Schochet1996} \cite{Serfaty2019}.

The viscous vortex was developed contemporaneously with the vortex blob model, and was first introduced in \cite{Chorin1973}. In this model, we introduce an independent Brownian noise, and obtain the SDE
\eq{\label{SDE1}
	\dr X^i = \frac1N\sum_{i=1}^N \alpha_i K(X^i-X^j)\,\dr t + \sqrt{2\nu}\sum_{i=1}^N\dr W^i_t,
} where $W^i_t$ are independent planar Brownian motions. This system was proved to have global solutions almost surely in \cite{Duerr1982}. However, we will work only with the associated Liouville equation,
\eq{\label{Liou1}
	\p_t\rho_N + \frac1N\sum_{i,j=1}^N \alpha_j K(x^i-x^j)\cdot\na_{x^i}\rho_N=\nu\sum_{i=1}^N\De_{x^i}\rho_N.
}
In \cite{Jabin-Wang2018}, Jabin and Wang prove a \emph{propagation of chaos} result in the case where $\alpha_i=1$, which we explain as follows. Given a solution $\rho_N$ to \eqref{Liou1} for $\alpha_i=1$ with $\rho_N\ge0$ and $\int\rho_N\,\dr X^N=1$, we assume for simplicity that $\rho_N(t,x^1,\ldots,x^N)$ is symmetric in the variables $x^1,\ldots,x^N$, and define the \emph{$k$-particle maringals} $\rho_{N,k}$ as 
\eq{
	\rho_{N,k}(x^1,\ldots,x^N)=\int_{\T^{2(N-k)}}\!{\rho_N(t,x^1,\ldots,x^N)}\,\dr x^{k+1}\ldots\,\dr x^N.
}

We first take $\omega>0$ to be a classical solution to \eqref{NS} normalized as a probability density. Denoting by $\bar\rho_N = \omega^{\ot N}$ the tensorized density $\bar\rho_N(x^1,\ldots,x^N) = \omega(x^1)\cdots\omega(x^N)$. In kinetic theory, \emph{chaos} refers to (near) independence of the distribution $\rho_N$, and so a \emph{propagation of chaos} result proves that $\rho_N(t,X)$ maintains (near) independence as $t$ increases. This notion is formulated in \cite{Jabin-Wang2018} using \emph{rescaled relative entropy}, which we define as 
\eq{
	\ca H_N(\rho_N\,|\,\bar\rho_N)(t) = \frac1N\int_{\T^{2N}}\!{\rho_N(t,X)\log\frac{\rho_N(t,X)}{\bar\rho_N(t,X)}}{\,\dr X}.}
 In this context, a propagation of chaos result is simply the statement that $\ca H_N(\rho_N^0\,|\,\bar\rho_N^0)\to 0$ as $N\to\iy$ implies $\ca H_N(\rho_N\,|\,\bar\rho_N)(t)\to0$ for all $t\ge 0$. More precisely, we have the following theorem from \cite{Jabin-Wang2018}.
 
\begin{thm}[\cite{Jabin-Wang2018}] For any classical solution $\omega>0$ of mass $1$ to \eqref{NS} and any entropy solution $\rho_N$ to \eqref{Liou1}, writing $\bar\rho_N = \omega^{\ot N}$, we have
\eq{
	\ca H_N(\rho_N\,|\,\bar\rho_N(t))\le e^{Mt}(\ca H_N(\rho_N^0\,|\,\bar\rho_N^0)(t)+\frac1N)
} for $M$ a constant dependent only on $\omega$ and $\frac1N\int\rho_N^0\log\rho_N^0$.
\end{thm}

 A corollary of this is the following mean field result (Corollary 1 of \cite{Jabin-Wang2018} applied to the Navier-Stokes system).

\begin{cor} Let $\omega\in C([0,\iy),C^2(\T^2))$ be a solution to vorticity Navier Stokes \eqref{NS}, with $\omega>0$, with normalization $\int_{\T^2}\omega=1$, and let $\rho_N$ be a strong solution to the Liouville equation \eqref{Liou1} with $\alpha_i=1$ for all $i$ for all $N$. Writing $\bar\rho_N = \omega^{\ot N}$ as before, assume $\ca H_N(\rho_N^0\, |\,\bar\rho_N^0)\to\iy$ as $N\to\iy$. Then
$$
	\ca H_N(\rho_N\,|\,\bar\rho_N)\xra{N\to\iy}0\quad\hbox{ over any } [0,T]
$$uniformly. Consequently, we have a strong propagation of chaos result
$$
	\norm{\rho_{N,k}-\bar\rho_k}_{L^\iy([0,T],L^1(\T^{2k}))}\xra{N\to\iy}0.
$$ Finally, when $\sup_N N\ca H_N(\rho_N\,|\,\bar\rho_N)=H<\iy$, then there exists $C = C(H,T,k,\sup_N\ca H_N(\rho_N^0\,|\,1))$ such that the quantitative estimate
$$
	\norm{\rho_{N,k}-\bar\rho_k}_{L^\iy([0,T],L^1(\T^{2k}))}\le\frac{C}{\sqrt N}
$$ holds.
\end{cor}
Whenever $(X^i(t))_{i=1}^N$ solves the SDE \eqref{SDE1} with $\alpha_i=1$ and $X^i(0)\sim\omega\,\dr x$ independently for all $1\le i\le N$, we have convergence
$$
	\frac1N\sum_{i=1}^N \de_{X^i(t)}\to\omega(t,x)\,\dr x
$$ almost surely in law, for all $t\ge0$.

We note, however, is that this is not a full approximation of 2D Navier-Stokes, since we have had to assume monotonicity $\omega>0$. As we shall prove in this paper, this condition can be relaxed by introducing a new limiting PDE, and then applying the large deviation type estimates used in \cite{Jabin-Wang2018}.

\subsection{Previous Results}

We compare our result mainly to that of \cite{Jabin-Wang2018}, where the force $K$ is allowed to belong to a very wide class, though particles are assumed to be identical, limiting its use to same sign results. Convergence results for mixed sign systems were obtained in \cite{Fournier-Hauray-Mischler2014} and \cite{Schochet1996} for the viscous and inviscid cases respectively, but in both cases without a rate. To the author's knowledge, this is the first quantitative propagation of chaos result which holds for a mixed sign viscous vortex model.

\section{Setup}

We now specialize to a \emph{two-species} mean-field result, where we wish to approximate the vorticity field $\omega$ by an empirical measure $$\frac1N\left[\alpha_+\sum_{i=1}^N\de_{X^i_+} -\alpha_-\sum_{i=1}^N\de_{X^i_-}\right]$$ for fixed $\alpha_\pm>0$. To do so, we define our SDE

\eq{\label{ourSDE}
	\dr X_\pm^i(t) = \frac{\alpha_+}{N}\sum_{j=1}^N K(X_\pm^i-X_+^j)\,\dr t
				-\frac{\alpha_-}{N}\sum_{j=1}^N K(X_\pm^i-X_-^j)\,\dr t + \sqrt{2\nu}\,\dr W^i_\pm(t)
				\qquad\hbox{ on }\T^2
} for $W^i_\pm(t)$ $2N$ independent Brownian motions. As before, however, we will deal not with the process itself, but with its probability density $\rho_N$, which will satisfy the Liouville equation
\eq{\label{Liouvillelong}
	\p_t\rho_N +\frac1N\sum_{i,j=1}^N\Big[&( \alpha_+K(x^i_+ - x^j_+) - \alpha_- K(x^i_+ - x^j_i))\cdot\na_{x_+^i}\rho_N \\
	&\quad+( \alpha_+K(x^i_i - x^j_+) - \alpha_- K(x^i_i - x^j_i))\cdot\na_{x_i^i}\rho_N \Big]
 = \nu\sum_{i=1}^N(\De_{x^i_+}\rho_N + \De_{x^i_-}\rho_N),
}
where we recall the convention that $K(0)=0$. To deal with the equation more effectively, we immediately make some notational simplifications. Firstly, we use $x^i = (x^i_+,x^i_-)$ for the rest of this paper, and define the \emph{weighted tensorized Biot-Savart kernel} to be 
\eq{
	\bar K(x,y) = \mat{ \alpha_+ K(x_+ - y_+) - \alpha_- K(x_+ - y_-) \\ 
					 \alpha_+ K(x_- - y_+) - \alpha_- K(x_- - y_-) } \in \R^4_{x_+,x_-},
} which allows us to dramatically simplify \eqref{Liouvillelong} as
\eq{\label{Liou}
	\p_t\rho_N + \frac1N\sum_{i,j=1}^N \bar K(x^i,x^j)\cdot\na_{x^i}\rho_N = \nu\sum_{i=1}^N\De_{x^i}\rho_N.
}

An immediate problem here is showing existence and entropy decay of solutions to \eqref{Liou}, given the roughness of the vector field $\bar K$. In general, the well-posedness theory for such problems can be hard, however, we do have a notion of \emph{entropy solutions}, which we will be using for the rest of the paper.

\begin{prop}\label{entest} Assume that $\int_{\T^{4N}}{\rho_N^0\log\rho_N^0}<\iy$, then \eqref{Liou} admits a (distributional) solution $\rho_N$ of time-dependent densities satisfying the entropy bound
\eq{
	\int_{\T^{4N}}\!{\rho_N(t,X)\log\rho_N(t,X)}{\,\dr X} + \nu\sum_{i=1}^N\int_0^t\!{\int_{\T^{4N}}\!{
		\frac{\abs{\na_{x_i}\rho_N}^2}{\rho_N}}{\,\dr X}}{\,\dr s}\le 
		\int_{\T^{4N}}\!{\rho_N^0(X)\log\rho_N^0(X)}{\,\dr X}
} for almost every $0\le t<\iy$. Moreover, if we assume $\rho_N^0$ to be symmetric, then for every $\phi\in L^2([0,T],W^{1,\iy}(\T^{2\ti 4}))$ with $\norm{\phi}_{L_T^2 W^{1,\iy}}\le1$, we have
\eq{
	\int_0^t\!\int_{\T^{2\ti 4}}\!{\phi(s,x^1,x^2)\bar K(x^1,x^2)\rho_{N,2}(s,x^1,x^2)}{,\dr x^1\,\dr x^2}{\,\dr s}
	\le C\pa{1+\frac{1}{\nu N}\int_{\T^{4N}}\!\rho_N^0\log\rho_N^0},
} so that $\bar K(x^1,x^2)\rho_{N,2}$ is well-defined, where $C$ depends only on $\alpha_\pm$.
\end{prop}
The proof of this is defered to the final section. We note additionally that a strong well-posedness theory exists for $\rho_N^0\in L^p(\T^{4N})$ for $p>2$, as can be checked by standard arguments, however we will not use this fact.

The key idea of this paper is to adapt the arguments of \cite{Jabin-Wang2018} for a \emph{tensorized} system, which we explain as follows. Since mixed sign solutions $\omega$ to \eqref{NS} cannot be naturally represented as probability densities, in order to use \cite{Jabin-Wang2018}, we are forced to consider positive and negative vorticity separately in our mean-field equation. More precisely, we introduce a \emph{tensorized Navier-Stokes} equation
\eq{\label{TNS}
	\p_t\bar\omega + \bar u\cdot\na\bar\omega = \nu\De\bar\omega,\quad \bar u= \bar{\ca K}\bar\omega
} where $\bar{\ca K}$ is the weighted tensorized Biot-Savart operator
\eq{
	\bar{\ca K}\bar\omega(x) = \int_{\T^4}\!{\bar K(x,y)\omega(y)}{\,\dr y}.
} 
We note that this is a generalization of vorticity Navier Stokes in a natural way. For any probability density $\phi:\T^2_{x_+}\ti\T^2_{x_-}\to\R$, we define projection functionals
\eqq{
	\mrm{pr}_\pm\phi(x_\pm) = \int_{\T^2}\!\phi(x_+,x_-)\,\dr x_\mp.
}
 Now, for any classical solution $\bar\omega$ to \eqref{TNS} (that is, $\omega\in C^1([0,\iy),C^0)\cap C^0([0,\iy),C^2)$), we obtain \emph{marginal distributions} $\omega_\pm = \mrm{pr}_\pm\bar\omega$, from which we may define the \emph{associated vorticity}
\eq{\label{TNScancel}
	\omega(t,x) = [\alpha_+\mrm{pr}_+-\alpha_-\mrm{pr}_-]\bar\omega(t,x).
} We have the following useful proposition, which motivates the introduction of \eqref{TNS}.
\begin{prop} Let $\bar\omega$ be a classical solution to \eqref{TNS} on $\T^{2\ti 2}$. Then $\omega =  [\alpha_+\mrm{pr}_+-\alpha_-\mrm{pr}_-]\bar\omega$ is a classical solution to Navier-Stokes.
\end{prop}
\begin{proof}The key step is to notice that $\bar u(x_+,x_-) = (u(x_+),u(x_-))$ (where we have omitted the time variable for concision). To see this, we simply compute
\eqq{
	\bar u(x_+,x_-) &= \int\!\int\!\mat{\alpha_+ K(x_+-x_+')-\alpha_- K(x_+-x_-')\\ \alpha_+ K(x_--x_+')-\alpha_- K(x_--x_-')}\bar\omega(x_+',x_-')\,\dr x_-'\,\dr x_+' \\
	&=\mat{\int\! K(x_+-x')\omega(x')\,\dr x' \\ \int\! K(x_--x')\omega(x')\,\dr x'} = \mat{u(x_+)\\u(x_-)}.
} Secondly, we compute that
\eqq{
	\int\! u(x_\pm)\cdot\na_{x_\pm}\bar\omega(x_+,x_-)\,\dr x_\pm = \int\!\na_{x_\pm}\cdot u(x_\pm)\bar\omega(x_+,x_-)\,\dr x_\pm=0.
} We can now apply the operator $\ca P= [\alpha_+\mrm{pr}_+-\alpha_-\mrm{pr}_-]$ to each term of tensorized Navier-Stokes \eqref{TNS}. First, we easily get that $\p_t\ca P\bar\omega = \ca P\p_t\bar\omega$. As for the Laplacian term, we can calculate that
$$
	(\mrm{pr}_\pm\De\bar\omega)(x_\pm) = \int\![\De_{x_+}+\De_{x_-}]\bar\omega(x_+,x_-)\,\dr x_\mp
	=\De_{x_\pm}\mrm{pr}_\pm\bar\omega(x_\pm)
$$ so that $\ca P\De\bar\omega = \De\ca P\bar\omega$. Finally, cancellation in \eqref{TNScancel} gives us that
$$
	\ca P(\bar u\cdot\na\bar\omega)(x) = \alpha_+\int\! u(x)\cdot\na_x\bar\omega(x,x_-)\,\dr x_-
		-\alpha_-\int\! u(x)\dot\na x\bar\omega(x_+,x)\,\dr x_+ = u(x)\cdot\na\omega(x)
$$ which proves that $\omega$ solves Navier Stokes.
\end{proof}

Having now motivated the introduction of \eqref{TNS}, we would like to show that classical solutions exist. The key fact we will use is that $\bar{\ca K}$ is a bounded operator $C^k(\T^4)\to C^k(\T^4)^4$, as can easily be checked. We note that \eqref{TNS} has a similar solution theory to the classical 2D Navier Stokes equation. In fact, we have the following general well-posedness result, which holds for any velocity operator $K:\bar\omega\mapsto\bar u$ with general boundedness properties.

\begin{thm}\label{nonlinnonlocthm} Fix some $k\ge0$, and let $K:C^0(\Td)\to C^0(\Td)^d$ be a velocity operator satisfying the conditions
\begin{enumerate}
	\item $\na\cdot(K\omega)=0$ weakly for all $\omega\in C^0(\Td)$ (incompressibility), and
	\item $K$ continuously maps $C^\ell(\Td)$ to $C^\ell(\Td)$ for all $0\le\ell\le k$ (boundedness).
\end{enumerate}

Then the PDE
\eq{\label{nonlinnonloc}
	\p_t\omega + u\cdot\na\omega = \nu\De\omega,\quad u =K\omega\qquad\hbox{ on }\Td
} with initial conditions $\omega_0\in C^k(\Td)$ is well-posed in $C([0,\iy),C^k(\Td))$ for any $k$, and furthermore, solutions $\omega$ lie in $C((0,\iy),C^k(\Td))$ even when we only have $\omega\in C^0(\Td)$.
\end{thm}
We note that this includes 2D Navier-Stokes in vorticity form, as well as our weighted tensorized generalization of 2D Navier Stokes. We prove this result in Section \ref{wellposedness}.

\section{Propagation of Chaos}

\begin{thm}\label{mainthm} Let $\bar\omega>0$ be a classical solution to tensorized Navier Stokes \eqref{TNS} normalized as $\int_{\T^4}\bar\omega\,\dr x=1$, and let $\rho_N$ be an entropy solution to the Liouville equation \eqref{Liou} defined for time $[0,T]$. Writing $\bar\rho_N = \bar\omega^{\ot N}$, we have the entropy estimate
\eq{
	\ca H_N(\rho_N\,|\,\bar\rho_N)(t)\le e^{Ct}\pa{\ca H_N(\rho_N^0\,|\,\bar\rho_N^0)+\frac1N},\quad 0\le t\le T
} where $C=C(\alpha_+,\nu,\inf\bar\omega_0,\norm{\bar\omega}_{L_T^\iy C^2})$.
\end{thm}

This theorem is an adaptation of Theorem 1 of \cite{Jabin-Wang2018}, for non-translation invariant systems.

We would now like to estimate the time evolution of the \emph{relative entropy} 
\eq{
	\HN(\rho_N|\bar\rho_N)(t)=\frac1N\int_{\T^{4N}}\!{\rho_N(t,X)\log\frac{\rho_N(t,X)}{\bar\rho_N(t,X)}}{\,\dr X},
} writing $X=(x^1,\ldots,x^N)$, which we calculate in the following lemma.
\begin{lem}\label{JW2018Lem2} Let $\rho_N$ be an entropy solution to the Liouville equation \eqref{Liou}, and let $\bar\omega\in C^1([0,\iy),C^0)\cap C^0([0,\iy),C^2)$ be a classical solution to tensorized Navier Stokes \eqref{TNS} with $\bar\omega>0$ and $\int_\Td\bar\omega=1$, then writing $\bar\rho_N = \bar\omega^{\ot N}$, we have
\eq{
	\HN(\rho_N|\bar\rho_N)(t)&\le\HN(\rho_N^0|\bar\rho_N^0)
	-\frac{1}{N^2}\sum_{i,j=1}^N\int_0^t\!{\int_{\T^{4N}}{\rho_N(K(x_i,x_j)-\bar{\ca K}\bar\rho(x_i))\cdot\na_{x_i}\log\bar\rho_N}{\,\dr X}}{\,\dr s} \\
	&\qquad-\frac\nu N\sum_{i=1}^N\int_0^t\!\int_{\T^{4N}}\!\rho_N\abs{\na_{x_i}\log\frac{\rho_N}{\bar\rho_N}}^2.
}
\end{lem}
\begin{proof} We directly apply the PDE \eqref{TNS} to see that $\log\bar\rho_N$ must solve
\eq{\label{eq0.10}
	\p_t\log\bar\rho_N&+\sum_{i=1}^N\frac1N\sum_{j=1}^N \bar K(x_i,x_j)\cdot\na_{x_i}\log\bar\rho_N
	=\sum_{i=1}^N\nu\frac{\De_{x_i}\bar\rho_N}{\bar\rho_N} 
	+\sum_{i=1}^N\pa{\frac1N\sum_{j=1}^N \bar K(x_i,x_j)-\bar{\ca K} \bar\rho(x_i)}\cdot\na_{x_i}\log\bar\rho_N.
} Since $\bar\rho_N$ is bounded below, $\log\bar\rho_N\in C^1([0,\iy)\ti\T^{4N})$ is a suitable test function in the Liouville equation \eqref{Liou} for $\rho_N$, so that
\eqq{
	\int_{\T^{4N}}\!{\rho_N\log\bar\rho_N}{\,\dr X}&=\int_{\T^{4N}}\!{\rho_N^0\log\bar\rho_N^0}{\,\dr X}
		+\int_0^t\!{\int_{\T^{4N}}\!{\rho_N\pa{\p_t\log\bar\rho_N
		+\frac1N\sum_{i,j=1}^N\bar K(x_i,x_j)\cdot\na_{x_i}\log\bar\rho_N}}{\,\dr X}}{\,\dr s} \\
	&\qquad-\nu\sum_{i=1}^N\int_0^t\!{\int_{\T^{4N}}\!{\na_{x_i}\log\bar\rho_N\cdot\na_{x_i}\rho_N}{\,\dr X}}{\,\dr s},
} but substituting in \eqref{eq0.10} gives us that
\eqq{
	\int_{\T^{4N}}\!{\rho_N\log\bar\rho_N}{\,\dr X}&=\int_{\T^{4N}}\!{\rho_N^0\log\bar\rho_N^0}{\,\dr X}
	+\sum_{i=1}^N\int_0^t\!{\int_{\T^{4N}}\!{\rho_N\pa{\frac1N\sum_{j=1}^N \bar K(x_i,x_j)-\bar{\ca K} \bar\rho(x_i)}\cdot \na_{x_i}\log\bar\rho_N}{\,\dr X}}{\,\dr s} \\
	&\qquad+\sum_{i=1}^N\int_0^t\!{\int_{\T^{4N}}\!{\pa{\nu\rho_N\frac{\De_{x_i}\bar\rho_N}{\bar\rho_N} - \nu\na_{x_i}\rho_N\cdot\frac{\na_{x_i}\bar\rho_N}{\bar\rho_N}}}{\,\dr X}}{\,\dr s}.
} Using the entropy dissipation we require from $\rho_N$ as an entropy solution, we get the bound
\eq{
	\HN(\rho_N|\bar\rho_N)(t)\le\HN(\rho_N^0|\bar\rho_N^0)
	-\frac{1}{N^2}\sum_{i,j=1}^N\int_0^t\!{\int_{\T^{4N}}\!{\rho_N\bar K(x_i,x_j)\cdot\na_{x_i}\log\bar\rho_N}{\,\dr X}}{\,\dr s}
	+\frac1N D_N
} for 
\eq{
	D_N&=\nu\sum_{i=1}^N\int_0^t\!\int_{\T^{4N}}\!\pa{-\rho_N\frac{\De_{x_i}\bar\rho_N}{\bar\rho_N}
	+\na_{x_i}\rho_N\cdot\frac{\na_{x_i}\bar\rho_N}{\bar\rho_N}-\frac{\abs{\na_{x_i}\rho_N}^2}{\rho_N}}\\
	&=-\nu\sum_{i=1}^N\int_0^t\!\int_{\T^{4N}}\!\pa{\rho_N\frac{\abs{\na_{x_i}\bar\rho_N}^2}{\bar\rho_N^2}
	-2\na_{x_i}\rho_N\cdot\frac{\na_{x_i}\bar\rho_N}{\bar\rho_N}+\frac{\abs{\na_{x_i}\rho_N}^2}{\rho_N}}\\
	&=-\nu\sum_{i=1}^N\int_0^t\!\int_{\T^{4N}}\!\rho_N\abs{\na_{x_i}\log\frac{\rho_N}{\bar\rho_N}}^2,
} which completes the proof.
\end{proof}
Since we would like to complete a Gronwall estmate, we must now analyze the error term in Lemma \ref{JW2018Lem2}. The first task is to find an upper bound involving the relative entropy $\HN$ – for this, we require the following elementary lemma.
\begin{lem}\label{JW2018Lem1} Fix probability densities $\rho_N$ and $\bar\rho_N$ on ${\T^{dN}}$, then for any $\Phi\in L^\iy({\T^{dN}})$, we have for all $\eta>0$ that
\eq{
	\int_{\T^{dN}}\!{\Phi\rho_N}{\,\dr X}\le\frac1\eta\pa{\HN(\rho_N|\bar\rho_N)+\frac1N\int_{\T^{dN}}\!{\bar\rho_N e^{N\eta\Phi}}{\,\dr X}}.
}
\end{lem}
\begin{proof} By rescaling $\Phi$, we may take $\eta=1$. Defining the probability density $$f=\frac1Z e^{N\Phi}\bar\rho_N,\quad Z = \int_{\T^{dN}}\!{\bar\rho_N e^{N\Phi}}{\,\dr X},$$ we may use convexity of the entropy to get that
\eqq{
	\frac1N\int_{\T^{dN}}\!{\rho_N\log f}{\,\dr X}=\int_{\T^{dN}}\!{\rho_N\log\bar\rho_N}{\,\dr X}
} which, by unraveling $\log f$, immediately gives us the result.
\end{proof}

We have now reduced the problem to a form which can be solved by methods from \cite{Jabin-Wang2018}. More precisely, we use the following probabilistic, large deviation type results from \cite{Jabin-Wang2018}. Both of these results rely on expanding the exponential as a power series, and use cancellations to obtain precise combinatorial estimates on the number of nonzero terms in the  expansion.

\begin{prop}[Theorem 3 of \cite{Jabin-Wang2018}]\label{JW2018Thm3} Suppose we have a density $\bar\rho\in L^1(\Td)$ with $\bar\rho\ge0$ and $\int_\Td\bar\rho=1$. Given a scalar function $\psi\in L^\iy$ with $\norm\psi_{L^\iy}\le1/2e$ and such that for any $z$, $\int_\Td\psi(z,x)\bar\rho(x)\,\dr x=0$, we have 
\eq{
	\int_{\T^{dN}}\!{\bar\rho_N\,\exp\pa{\frac1N\sum_{i,j=1}^N\psi(x_1,x_i)\psi(x_1,x_j)}}{\,\dr X}\le 
	C=2\pa{1+\frac{10\alpha}{(1-\alpha)^3}+\frac{\beta}{1-\beta}}
} where $\bar\rho_N=\bar\rho^{\ot N}$, and $$\alpha=(e\norm{\psi}_{L^\iy})^4<1,\quad\beta=(\sqrt{2e}\norm\psi_{L^\iy})^4<1.$$
\end{prop}

\begin{prop}[Theorem 4 of {[JW2018]}]\label{JW2018Thm4} Let $\bar\rho\in L^1(\Td)$ be a density. Consider any $\phi(x,z)\in L^\iy$ with $$\ga:=C\norm{\phi}_{L^\iy}^2<1$$ for some universal constant $C$. Then if $\phi$ satisfies the cancellation conditions
\eq{
	\int_\Td\!{\phi(x,z)\bar\rho(x)}{\,\dr x}=0\quad\forall z,\qquad\int_\Td\!{\phi(x,z)\bar\rho(z)}{\,\dr z}=0\quad\forall x,
} then we have
\eq{
	\int_{\T^{dN}}\!{\bar\rho_N\exp\pa{\frac1N\sum_{i,j=1}^N\phi(x_i,x_j)}}{\,\dr X}\le\frac{3}{1-\ga}<\iy.
}
\end{prop}

We can now make use of these probabilistic results to get the following estimate on the mean-field correction.

\begin{lem}\label{JW2018Lem3} Assume $\bar\rho\in C^2(\T^4)$ with $\inf\bar\rho>0$, then we have the bound
\eq{
	-\frac{1}{N^2}\sum_{i,j=1}^N\int_{\T^{4N}}\!{\rho_N(\bar K(x_i,x_j)-\bar{\ca K}\bar\rho(x_i))\cdot\na_{x_i}\log\bar\rho_N}{\,\dr X}
	&\le\frac{\nu}{4N}\sum_{i=1}^N\int_{\T^{4N}}\!{\rho_N\abs{\na_{x_i}\log\frac{\rho_N}{\bar\rho_N}}^2}{\,\dr X} \\ &\quad+ C\pa{\HN(\rho_N|\bar\rho_N)+\frac1N}
} for $C=C(\nu,\alpha_\pm,\norm{\bar\rho}_{C^2},\inf\bar\rho)$.
\end{lem}
\begin{proof} We now use the fact that the Biot-Savart kernel $K$ on $\T^2$ can be represented as $K(x) = \na_x\cdot V_0(x)$, where $V_0\in L^\iy(\T^2)$ is a matrix-valued function \cite{Jabin-Wang2018}. In fact, for the periodic Biot-Savart kernel, this matrix can be made semi-explicit. We note that on $\R^2$,
\eq{
	\p_i\pa{\frac{1}{2\pi}\frac{x_ix_j}{\abs x^2}} = \frac{1}{2\pi}\frac{x_j}{\abs{x}^2},
} which allows us to define
\eq{
	V_0(x) = \frac{1}{2\pi}\frac{x\ot x^\perp}{\abs x^2} + \Gamma(x),\qquad x\in [-1/2-\ve,1/2+\ve]
} for a matrix-valued smooth correction $\Gamma\in C^\iy$. We similarly define
\eq{
	V(x,x') = \mat{\alpha_+ V_0(x_+-x_+')-\alpha_- V_0(x_+-x_-') & 0 \\ 0 & 
				\alpha_+ V_0(x_--x_+')-\alpha_- V_0(x_--x_-') }
} where the block matrices represent $x_+$ and $x_-$ coordinates respectively. We additionally define
\eq{
	\upsilon(x) = \mat{-\alpha_- V_0(x_+-x_-) & 0 \\ 0 & \alpha_+ V_0(x_--x_+)}.
} We see easily that $V\in L^\iy(\T^4\ti\T^4)$, and $\upsilon\in L^\iy(\T^4)$, and that $\na_x\cdot V(x,x') = \bar K(x,x')$ and $\na_x\cdot\upsilon(x) = K(x,x)$. We can then rewrite
\eqq{
&-\frac{1}{N^2}\sum_{i,j=1}^N\int_{\T^{4N}}{\rho_N(\bar K(x_i,x_j)-\bar{\ca K}\bar\rho(x_i))\cdot\na_{x_i}\log\bar\rho_N}{\,\dr X} \\
&\qquad= -\frac{1}{N^2}\sum_{\alpha,\beta}\sum_{i=1}^N\int_{\T^{4N}}{\p_{x^i_\beta}\pa{
	\sum_{j\ne i}V_{\alpha\beta}(x_i,x_j)+\upsilon(x_i)-N\,T_{V_{\alpha \beta}}\bar\rho(x_i)}\frac{\rho_N}{\bar\rho_N}\p_{x^i_\alpha}\bar\rho_N}{\,\dr X}
} where $T_{V_{\alpha\beta}}$ is the integral operator obtained by integration in the second variable. We can then integrate by parts, and get that this is equal to
\eq{
	&=\frac{1}{N^2}\sum_{\alpha,\beta}\sum_{i=1}^N\int_{\T^{4N}}\!{\pa{\sum_{j\ne i}V_{\alpha\beta}(x_i,x_j)+\upsilon(x_i)-NT_{V_{\alpha \beta}}\bar\rho(x_i)}\frac{\rho_N}{\bar\rho_N}\p^2_{x^i_\alpha x^i_\beta}\bar\rho_N}{\,\dr X} \\
	&\qquad+\frac{1}{N^2}\sum_{\alpha,\beta}\sum_{i=1}^N\int_{\T^{4N}}\!{\pa{\sum_{j\ne i}V_{\alpha\beta}(x_i,x_j)+\upsilon(x_i)-NT_{V_{\alpha \beta}}\bar\rho(x_i)}\p_{x^i_\beta}\frac{\rho_N}{\bar\rho_N}\p_{x^i_\alpha}\bar\rho_N}{\,\dr X} \\
	&=\frac{1}{N^2}\sum_{i,j=1}^N\int_{\T^{4N}}\!{(V(x_i,x_j)-T_V\bar\rho(x_i)):\na_{x_i}\bar\rho_N\ot\na_{x_i}
		\frac{\rho_N}{\bar\rho_N}}{\,\dr X} \\
	&\qquad + \frac{1}{N^2}\sum_{i,j=1}^N\int_{\T^{4N}}\!{\rho_N\,(V(x_i,x_j)-T_V\bar\rho(x_i)):
		\frac{\na_{x_i}^2\bar\rho_N}{\bar\rho_N}}{\,\dr X} =: A+B
} in tensor form, where we use the convention that $V(x,x):=\upsilon(x)$. We now independently treat the quantities $A$ and $B$. To bound $A$, we first use Cauchy Schwartz and Young's inequality to get that
\eq{\label{firstAbound}
	A &\le \frac{\nu}{4N}\sum_{i=1}^N\int_{\T^{4N}}\!{\frac{\bar\rho_N^2}{\rho_N}\abs{\na_{x_i}\frac{\rho_N}{\bar\rho_N}}^2}{\,\dr X}
	+\frac{4}{N\nu}\sum_{i=1}^N\int_{\T^{4N}}\!{\abs{\frac1N\sum_{j=1}^N V(x_i,x_j)-T_V\bar\rho(x_i)}^2
	\abs{\frac{\na_{x_i}\bar\rho_N}{\bar\rho_N}}^2\rho_N}{\,\dr X} \\
	&\le \frac{\nu}{4N}\sum_{i=1}^N\int_{\T^{4N}}\!{\rho_N\abs{\na_{x_i}\log\frac{\rho_N}{\bar\rho_N}}^2}{\,\dr X}
	+\frac{4\norm{\bar\rho}_{C^1}^2}{N\nu(\inf\bar\rho)^2}\sum_{i=1}^N
		\int_{\T^{4N}}\!{\abs{\frac1N\sum_{j=1}^NV(x_i,x_j)-T_V\bar\rho(x_i)}^2\rho_N}{\,\dr X}
} where we take the Hilbert-Schmidt (i.e. entrywise $\ell^2$) norm of the matrix-valued terms. We now let $\eta>0$ vary and apply Lemma \ref{JW2018Lem1} to 
\eqq{
	\Phi = \eta^2\pa{\frac1N\sum_{j=1}^NV_{\alpha\beta}(x_i,x_j)-T_{V_{\alpha\beta}}\bar\rho(x_i)}^2
} and get
\eq{\label{secondAbound}
&\frac1N\sum_{i=1}^N\int_{\T^{4N}}\!{\abs{\frac1N\sum_{j=1}^NV(x_i,x_j)-T_V\bar\rho(x_i)}^2\rho_N}{\,\dr X}\\
&\qquad\le\frac{4^2}{\eta^2}\HN(\rho_N|\bar\rho_N)
 + \frac{1}{N^2\eta^2}\sum_{i=1}^N\sum_{\alpha,\beta=1}^4
\log\int_{\T^{4N}}\!{\bar\rho_N e^{N\eta^2(\frac1N\sum_jV_{\alpha\beta}(x_i,x_j)-T_{V_{\alpha\beta}}\bar\rho(x_i))^2}}{\,\dr X}.
} But by symmetry, we get that
\eq{
	&\frac1N\sum_{i=1}^N\log\int_{\T^{4N}}\!{\bar\rho_N e^{N\eta^2(\frac1N\sum_jV_{\alpha\beta}(x_i,x_j)-T_{V_{\alpha\beta}}\bar\rho(x_i))^2}}{\,\dr X} \\
	 &\qquad=\log\int_{\T^{4N}}\!{\bar\rho_N e^{N\eta^2(\frac1N\sum_jV_{\alpha\beta}(x_1,x_j)-T_{V_{\alpha\beta}}\bar\rho(x_1))^2}}{\,\dr X}.
} We would now like to apply Proposition \ref{JW2018Thm3}. To do so, we define $\psi(z,x)=\eta V_{\alpha\beta}(z,x)-\eta T_{V_{\alpha\beta}}\bar\rho(z)$ and let $\eta=1/(4e\norm V_{L^\iy})$, noting that $\norm\psi_{L^\iy}<1/4e$ and that by definition of $T_{V_{\alpha\beta}}$, we have cancellation $\int\bar\rho(x)\psi(z,x)\,\dr x=0$ for all $z$. But since
\eqq{
	N\eta^2\pa{\frac1N\sum_{j=1}^N V_{\alpha\beta}(x_1-x_j)-T_{V_{\alpha\beta}}\bar\rho(x_1)}^2
	=\frac1N\sum_{j_1,j_2=1}^N\psi(x_1,x_{j_1})\psi(x_1,x_{j_2}),
} we may apply Proposition \ref{JW2018Thm3} to get that
\eqq{
	\int_{\T^{4N}}\!{\bar\rho_Ne^{N\eta^2(\frac1N\sum_jV_{\alpha\beta}(x_i,x_j)-T_{V_{\alpha\beta}}\bar\rho(x_i))^2}}{\,\dr X}\le C
} for an absolute constant $C$. But combining this bound with the bounds from \eqref{firstAbound} and \eqref{secondAbound} gives us
\eq{\label{finalAbound}
	A\le\frac{\nu}{4N}\sum_{i=1}^N\int_{\T^{4N}}\!{\rho_N\abs{\na_{x_i}\log\frac{\rho_N}{\bar\rho_N}}^2}{\,\dr X}
	+C\frac{\norm{\bar\rho}^2_{C^1}\norm{V}_{L^\iy}^2}{\nu(\inf\bar\rho)^2}\pa{\HN(\rho_N|\bar\rho_N)+\frac1N}.
 } 
 
 Next, we work on obtaining a bound on $B$. Defining
 \eq{
 	\phi(x,z)=(V(x,z)-T_V\bar\rho(x)):\frac{\na_x^2\bar\rho(x)}{\bar\rho(x)}
 } lets us rewrite
\eq{
	B = \frac{1}{N^2}\int_{\T^{4N}}\!{\rho_N\phi(x_i,x_j)}{\,\dr X}
	\le\frac1\eta\HN(\rho_N|\bar\rho_N)+\frac{1}{N\eta}\int_{\T^{4N}}\!{\bar\rho_N e^{\frac1N\sum_{i,j}\eta\phi(x_i,x_j)}}{\,\dr X}
} by application of Lemma \ref{JW2018Lem1} to $\Phi=\frac{1}{N^2}\sum_{i,j}\eta\phi(x_i,x_j)$. Meanwhile, integration by parts gives us that
\eq{
	\int_\Td\!{(V(x,z)-T_V\bar\rho(x)):\frac{\na_x^2\bar\rho(x)}{\bar\rho(x)}\bar\rho(x)}{\,\dr x}
	=\int_\Td\!{(\na_x\cdot \bar K(x,z)-\na_x\cdot\bar{\ca K}\bar\rho(x))\bar\rho(x)}{\,\dr x}=0
} so that $\phi$ satisfies the double cancellation property of Proposition \ref{JW2018Thm4}. We also note from the definition of $\phi$ that
\eq{
	\norm{\phi}_{L^\iy}\le \frac{\norm{V}_{L^\iy}^2}{\inf\bar\rho}\norm{\bar\rho}_{C^2}.
} Choosing
\eq{
	\eta=\frac{\inf\bar\rho}{C\norm V_{L^\iy}\norm{\bar\rho}_{C^2}}
} means we can use Proposition \ref{JW2018Thm4} to bound
\eq{
	\int_{\T^{4N}}\!{\bar\rho_N e^{\frac1N\sum_{i,j}\eta\phi(x_i,x_j)}}{\,\dr X}\le C
} for an absolute constant $C$, giving the bound
\eq{\label{finalBbound}
	B\le C\frac{\norm V_{L^\iy}\norm{\bar\rho}_{C^2}}{\inf\bar\rho}\pa{\HN(\rho_N|\bar\rho_N)+\frac1N}.
} We now combine this with the bound on $A$ in \eqref{finalAbound}, which concludes the proof.
\end{proof}

We are now able to prove the main theorem of this report as an easy corollary of the preceding work.

\begin{proof}[Proof of Theorem \ref{mainthm}] Taking Lemma \ref{JW2018Lem2} and applying Lemma \ref{JW2018Lem3} allows us to write
\eq{
	\HN(\rho_N|\bar\rho_N)(t)\le\HN(\rho_N^0|\bar\rho_N^0) + CM_K\int_0^t\!{\pa{\HN(\rho_N|\bar\rho_N)(s)+\frac1N}}{\,\dr s}
} using the formula for $M_K$ obtained in Proposition \ref{JW2018Thm3}, but then by Gronwall's lemma, we get that 
\eq{
	\HN(\rho_N|\bar\rho_N)(t)\le e^{M_K t}\pa{\HN(\rho_N^0|\bar\rho_N^0)+\frac1N}
} which concludes the proof.
\end{proof}

\section{Well-posedness Theory}\label{wellposedness}

We now prove the well-posedness theorem for our generalized convection diffusion equation,
\eq{
	\p_t\omega+u\cdot\na\omega = \nu\De\omega,\quad u = K\omega\qquad\hbox{ on }\Td.
}

A local existence result follows essentially immediately, by usual fixed point methods. 
\begin{lem}\label{locintime} The problem \eqref{nonlinnonloc} set up as in Theorem \ref{nonlinnonlocthm} is locally well-posed.
\end{lem}
We first note that the rescaling $t\nu = t'$ reduces to \eqref{nonlinnonloc} with $\nu=1$ and the rescaled flow operator $\frac1\nu K$, so we assume $\nu=1$ in the follows. Now, in order to apply fixed point methods, we use that $K\omega$ is divergence-free and rewrite \eqref{nonlinnonloc} as the integral equation
\eq{
	\omega(t) = e^{ t\De}\omega_0-\int_0^t\!{e^{(t-s)\De}\na\cdot(u\omega)(s)}{\,\dr s},\quad u = K\omega.
}
The meat of the proof now comes from estimates on $e^{t\De}\na$. We are helped here by the fact that $e^{t\De}$ is a convolutional operator whose kernel can be written explicity, as 
\eq{\label{originalheat kernel}
	G_t(x)=\frac{1}{(4\pi t)^{d/2}}e^{-\abs x^2/4t}
} on $\Rd$, and as the periodization of the above on $\Td$. Simple computations of Gaussian integrals then give us that
\eq{
	\norm{\na G_t}_{L^1(\Rd)} = \frac{1}{(4\pi t)^{d/2}}\int_\Rd\!{\frac{\abs x}{2t}e^{-\abs x^2/4t}}{\,\dr x}
	=\frac{\Ga\pa{frac{d+1}{2}}}{\Ga\pa{frac d2}} t^{-\frac12}.
} This then gives us the bound
\eq{
	\norm{\na G_t}_{L^1(\Td)} \le C t^{-\frac12}
} for a dimensional constant $C$, where by slight abuse of notation we write $G_t$ both for \eqref{originalheatkernel} and its periodization. Defining the Picard iteration
\eq{
	\ca R\omega(t) = e^{t\De}\omega_0 - \int_0^t\!{e^{(t-s)\De}\na\cdot(u\omega)(s)}{\,\dr s},\quad u=K\omega
} for $\omega\in C([0,T],C^k(\Td))$, we get that 
\eq{
	\na^\ell\ca R\omega(t) = e^{t\De}\na^\ell\omega_0-\int_0^t\!{e^{(t-s)\De}\na\cdot\na^\ell(u\omega)(s)}{\,\dr s},\quad u = K\omega,
} so, writing $M$ for the operator norm of $K:C^k(\Td)\to C^k(\Td)$, we get
\eq{
	\norm{\na^\ell\ca R\omega(t)}_{C^0} &\le \norm{\na^\ell\omega_0}_{C^0} + C\int_0^t\!{(t-s)^{-\frac12}\norm{\na^\ell(u\omega)(s)}_{C^0}}{\,\dr s} \\
	&\le \norm{\na^\ell\omega_0}_{C^0} + CM\sqrt T \norm{\omega}_{C^0([0,T],C^k)}^2
} for $C$ dependent on $d$ and $k$. Defining $X_T^k$ as the space $C^0([0,T],C^k(\Td))$ gives us the estimate
\eq{
	\norm{\ca R\omega}_{X_T^k}\le\norm{\omega_0}_{C^k}+CM\sqrt T\norm{\omega}_{X_T^k}^2.
} But then we see that for $$T\le\frac{1}{16C^2 M^2\norm{\omega_0}_{C^k}^2},$$ this maps the ball $\overline{B_{X_T^k}}(0,2\norm{\omega_0}_{C^k})$ to itself. To show that $\ca R$ is a contraction, we similarly estimate that for $\ell\le k$ and $t\in[0,T]$,
\eq{
	\norm{\na^\ell(\ca R\omega-\ca R\omega')(t)}_{C^0} &\le C\int_0^t\!{(t-s)^{-\frac12}\norm{\na^\ell(u\omega)(s)-\na^\ell(u'\omega')(s)}_{C^\ell}}{\,\dr s},\qquad u = K\omega,\ u'=K\omega' \\
	&\le CM\sqrt T (\norm{\omega}_{X_T^k}+\norm{\omega}_{X_T^k})\norm{\omega-\omega'}_{X_T^k},
} by calculations of type 
$$\p^\alpha u\,\p^\beta\omega-\p^\alpha u'\p^\beta\omega' = 
	\p^\beta\omega(\p^\alpha u -\p^\alpha u') + \p^\alpha u'(\p^\beta\omega-\p^\beta\omega'),$$ giving us the estimate
\eq{
	\norm{\ca R\omega-\ca R\omega'}_{X_T^k}\le CM\sqrt T(\norm{\omega}_{X_T^k}+\norm{\omega}_{X_T^k})\norm{\omega-\omega'}_{X_T^k}.
}	
Therefore, if we additionally impose that $$T<\frac{1}{4 C^2 M^2}$$ for a possibly bigger $C$, we have that $\ca R$ is a contraction on $\overline{B_{X_T^k}}(0,2\norm{\omega_0}_{C^k})$ and thus has a unique solution in that space. By continuity of solutions $t\mapsto \omega(t)\in C^k$ (and continuous dependence of $\ca R$ on $\omega_0\in C^k$), we have local-well posedness, valid up to blowup of the $C^k$ norm, which concludes the proof of Lemma \ref{locintime}.

For global-in-time well posedness, we make the critical observation that as a transport equation, \eqref{nonlinnonloc} remains bounded in $C^0$ for the full interval of existence of solutions.

\begin{prop} The equation \eqref{nonlinnonloc} satisfies a maximal principle, that is, for any solution $\omega\in C([0,T],C^0)$ to \eqref{nonlinnonloc} with initial condition $\omega_0\in C^0$, whenever $A\le\omega_0\le B$ for constants $A,B\in\R$, we have $A\le\omega(t)\le B$ for all $t$. Consequently, we have global existence of solutions to \eqref{nonlinnonloc} in $C^0$.
\end{prop}
\begin{proof} We exploit the local existence theory in $C^k$ to obtain classical solutions, and prove the maximum principle there. Assuming $\omega_0\in C^3$, we get by boundedness of $K:C^3\to C^3$ that $\omega\in C^1([0,T],C^1)\cap C^0([0,T],C^3)$. We use the usual approximation trick, letting $\omega' = \omega=\omega-\ve t$ and prove a maximum principle for $\omega'$. Fixing a $0\le t\le T$, we let $x\in\Td$ maximimize $\omega'(\cdot,t)$, so that $\nabla\omega(x)=0$ and $\De\omega(x)\ge 0$. However, since $$\p_t\omega'+u'\cdot\na\omega'<\nu\De\omega',$$ we get that $\max_x\omega'(x,t)$ is strictly decreasing for $t\in [0,T]$, and therefore $\max_x\omega(x,t)$ is decreasing for all $t$. The minimum principle is proved similarly. Density $C^3\subset C^0$ then proves the claims.
\end{proof}

Finally, in order to have global well-posedness in $C^k$, we use a bootstrap argument to improve regularity over time. In fact, just as we expect for 2D Navier Stokes, we have immediate smoothing of solutions due to the viscosity term, which we prove estimating $\norm{\na^{\ell+1}\omega(t)}_{C^0}$ as a function of $\norm{\na^\ell\omega_0}_{C^0}$. To do this, we assume that we have a solution $\omega\in C([0,T^*],C^{k+1})$, and writing 
\eq{
	\na^{\ell+1}\omega(t)= e^{t\De}\na(\na^\ell\omega_0)-\int_0^t\!{e^{(t-s)\De}\na\cdot\na^{\ell+1}(u\omega)(s)}{\,\dr s},
} for any $\ell\le k$, we may immediately estimate
\eq{
	\norm{\na^{\ell+1}\omega(t)}_{C^0}&\le C t^{-\frac12}\norm{\na^\ell\omega_0}_{C^0} + 
			C\int_0^t\!{(t-s)^{-\frac12}\norm{\na^{k+1}(u\omega)(s)}}{\,\dr s} \\
			&\le C t^{-\frac12}+C\int_0^t\!{(t-s)^{-\frac12}\norm{\omega(s)}_{C^k}\norm{\omega(s)}_{C^{k+1}}}{\,\dr }
} so that 
\eq{\label{ineq1}
	\norm{\omega(t)}_{C^{k+1}}\le C t^{-\frac12}\norm{\omega_0}_{C^k}+C\int_0^t\!{(t-s)^{-\frac12}\norm{\omega(s)}_{C^k}\norm{\omega(s)}_{C^{k+1}}}{\,\dr s}.
} 

When $\ell=0$, since $\norm{\omega(s)}_{C^0}\le\norm{\omega_0}_{C^0}$, we can iterate the inequality \eqref{ineq1} to get
\eq{
	\norm{\omega(t)}_{C^1}&\le C't^{-\frac12}+C'\int_0^t\!{(t-s)^{-\frac12}\norm{\omega(s)}_{C^1}}{\,\dr s} \\ 
		&\le C' t^{-\frac12}+C'^2\int_0^t\!{(t-s)^{-\frac12}s^{-\frac12}}{\,\dr s} + C'^2\int_0^t\!{\int_0^s\!{(t-s)^{-\frac12}(s-r)^{-\frac12}\norm{\omega(r)}_{C^1}}{\,\dr r}}{\,\dr s} \\
		&\le C' t^{-\frac12} + \pi C'^2 + \pi C'^2\int_0^t\!{\norm{\omega(r)}_{C^1}}{\,\dr r}
} for $C'$ a constant now dependent on $\norm{\omega_0}_{C^0}$ also, and where we used the Beta function identity
$$ \int_0^t\!{(t-s)^{-\frac12}s^{-\frac12}}{\,\dr s}=\pi$$
twice. This then allows us to conclude by Gronwall that
\eq{
	\norm{\omega(t)}_{C^1} \le C(1+t^{-\frac12})+C\int_0^t\!{(1+s^{-\frac12})e^{C(t-s)}}{\,\dr s}
	\le C(1+t^{-\frac12}+e^{Ct}).
} 
Procluding blowup for the higher $C^{\ell+1}$ norms proceeds in a similar way, and by induction. Fixing any $0<\ve<T'<\iy$, we show by induction that there exists a constant $C = C(\ve,T,\norm{\omega_0}_{C^0},d,\ell)$ such that $\norm{\omega(t)}_{C^{\ell+1}}\le C$ for any $t\in [\ve,T]$ such that a solution is defined. Assuming that the statement was already proved for $\norm{\omega(t)}_{C^\ell}$, we write $\ve'=\ve/2$ and get that
\eq{
	\norm{\omega(t)}_{C^{\ell+1}}\le C\norm{\omega(\ve')}_{C^\ell}(t-\ve')^{-\frac12} + C\int_{\ve'}^t\!{(t-s)^{-\frac12}\norm{\omega(s)}_{C^{\ell+1}}}{\,\dr s}
} which implies
\eq{
	\norm{\omega(t)}\le C(1+(t-\ve')^{-\frac12} + e^C(t-\ve')) \le C
} but then this proves the claim, and thus we conclude the proof of Theorem \ref{nonlinnonlocthm}.

\section{Entropy Estimate}

We now prove the entropy estimate to the Liouville equation
\eq{\label{Liou2}
	\p_t\rho + \frac1N\sum_{i,j=1}^N \bar K(x^i,x^j)\cdot\na_{x^i}\rho = \nu\sum_{i=1}^N\De_{x^i}\rho.
} 
We follow the methods of Proposition 1 in \cite{Jabin-Wang2018}.

\begin{proof}[Proof of Lemma \label{entest}] We note that, replacing $\bar K$ by a (divergence-free) regularized field $\bar K_\ve\in C^\iy(\T^4)$ and solving for

\eq{
	\p_t\rho_\ve + \frac1N\sum_{i,j=1}^N \bar K_\ve(x^i,x^j)\cdot\na_{x^i}\rho_\ve = \nu\sum_{i=1}^N\De_{x^i}\rho_\ve.
} gives us the entropy estimate 
\eq{
	\int_{\T^{4N}}\!{\rho_\ve(t,X)\log\rho_\ve(t,X)}{\,\dr X} + \nu\sum_{i=1}^N\int_0^t\!{\int_{\T^{4N}}\!{
		\frac{\abs{\na_{x_i}\rho_\ve}^2}{\rho_\ve}}{\,\dr X}}{\,\dr s}\le 
		\int_{\T^{4N}}\!{\rho_\ve^0(X)\log\rho_\ve^0(X)}{\,\dr X}
} for free, since
\eq{
	\frac{\dr}{\dr t}\int\!\rho_\ve\log\rho_\ve\,\dr X &= \int\!\p_t\rho_\ve[\log\rho_\ve+1]\,\dr X \\
	&=-\int\!\frac1N\sum_{i,j=1}^N\bar K_\ve(x^i,x^j)\cdot\na_{x^i}\rho_\ve[\log\rho_\ve+1]
		+ \int\!\nu\sum_{i=1}^N\De_{x^i}\rho_\ve[\log\rho_\ve+1]\,\dr X \\
	&= \frac1N\sum_{i,j=1}^N\int\!\bar K_\ve(x^i,x^j)\cdot\na_{x^i}\rho_\ve - \nu\sum_{i=1}^N\int\!\frac{\abs{\na_{x^i}\rho_\ve}^2}{\rho_\ve}\,\dr X \\
	&= - \nu\sum_{i=1}^N\int\!\frac{\abs{\na_{x^i}\rho_\ve}^2}{\rho_\ve}\,\dr X 
} where we twice integrate by parts in the convection term, and use that $\na_x\cdot\bar K_\ve(x,x')=0$.

Now, since $\rho_\ve$ are all probability densities, we can use that
\eq{
	\int\!\frac{\abs{\na\rho_\ve}}{\rho_\ve}\rho_\ve\,\dr X\le\pa{\int\!\frac{\abs{\na\rho_\ve}^2}{\rho_\ve^2}\rho_\ve\,\dr X}^{1/2}
} which gives us that $\rho_\ve$ is a bounded sequence in $C([0,T],W^{1,1}(\T^{4N}))$. Thus, we can extract a strongly convergent subsequence, which wtill satisfy the entropy estimate
\eq{
	H(\rho)(t) + \nu\int_0^t\!\int_{\T^{4N}}\!\frac{\abs{\na\rho}^2}{\rho}\,\dr x\,\dr s\le H(\rho^0)
} for all $t\in[0,T]$.

To prove the second estimate, we just note that by convexity \cite{MischlerMouhot}, we have 
\eq{
	\int_{\T^{2\ti4}}\!\frac{\abs{\na_{x_1}\rho_{N,2}}^2}{\rho_{N,2}}\,\dr x_1\,\dr x_2\le
		\int_{\T^{4N}}\!\frac{\abs{\na_{x_1}\rho_N}^2}{\rho_N}\,\dr X.
} Using the representation $K(x,x')=\na_x\cdot V(x,x')$, we get
\eq{
	\int_{\T^{2\ti 4}}\!\bar K(x_1,x_2)\phi(x_1,x_2)\rho_{N,2}(x_1,x_2)\,\dr x_1\,\dr x_2
	&=-\int_{\T^{2\ti4}}\!\bar V(x_1,x_2)(\phi\na_{x_1}\rho_{N,2}+\na_{x_1}\phi\rho_{N,2}\,\dr x_1\,\dr x_2\\
	&\le\norm{\phi}_{C^1}\norm{V}_{L^\iy}
		\pa{\int_{\T^{2\ti 4}}\!\frac{\abs{\na_{x_1}\rho_{N,2}}^2}{\rho_{N,2}}\,\dr x_1\,\dr x_2}^{\frac12}
} which by integrating and applying Holder and then Young and then the previous part of the proposition, proves the claim.
\end{proof}

\section{Conclusion}

We have presented a new approach to the mixed-sign point vortex system, by using as an intermediary a tensorized form of Navier-Stokes and applying large deviation estimates from to the mean field error. 
Our result shows that a mixed sign point vortex model is an accurate model for Navier-Stokes on the torus, and that we have convergence of the marginal particle distributions at the optimal rate of $O(N^{-1/2})$, as in \cite{Jabin-Wang2018}. 

\subsection{Further questions}
 Since the constant in our estimate blows up as $\nu^{-1}$ in the vanishing viscosity limit, we are unable to treat this case. The vanishing viscosity case is covered in \cite{Jabin-Wang2018} but with the stricter condition that the force $K$ obeys $\abs x K(x)\in L^\iy$, which does not allow for our tensorized kernel. Previous mean field results for the mixed sign inviscid case use either initial data on a grid \cite{Goodman-Hou-Louwengrub1990} or use compactness methods and stochastic initial data chosen according to specific schemes to be well distributed \cite{Schochet1996}.
 
 Additionally, for physical and numerical reasons, we would like to generalize the two-species model to a multi-species or a stochastically chosen vortex strength model, as in \cite{Fournier-Hauray-Mischler2014}. Indeed, early results on convergence of the point vortex model generally used nonidential vortex strenghts \cite{Hald1978} \cite{Goodman-Hou-Louwengrub1990}. While the multi-species model is a straightforward rewriting of our arguments however, a stochastic vortex strength model would require a somewhat more careful analysis.

\printbibliography

\end{document}